\documentclass[13pt]{extarticle}
\usepackage{amsfonts}
\usepackage{amsthm}
\usepackage{amsmath}
\usepackage{amscd}
\usepackage[latin2]{inputenc}
\usepackage{t1enc}
\usepackage[mathscr]{eucal}
\usepackage{indentfirst}
\usepackage{graphicx}
\usepackage{graphics}
\usepackage{pict2e}
\usepackage{epic}
\usepackage[margin=3.6cm]{geometry}
\usepackage{epstopdf}

\theoremstyle{plain}
\newtheorem{Th}{Theorem}[section]
\newtheorem{Lemma}[Th]{Lemma}

 \theoremstyle{definition}
\newtheorem{Def}[Th]{Definition}

\newtheorem{?}[Th]{Problem}

\begin{document}

\title{Results on Pattern Avoidance Games}

\author{P.A. CrowdMath}

\maketitle

\begin{abstract} A zero-one matrix $A$ contains another zero-one matrix $P$ if some submatrix of $A$ can be transformed to $P$ by changing some ones to zeros. $A$ avoids $P$ if $A$ does not contain $P$. The Pattern Avoidance Game is played by two players. Starting with an all-zero matrix, two players take turns changing zeros to ones while keeping $A$ avoiding $P$. We study the strategies of this game for some patterns $P$. We also study some generalizations of this game. 
\end{abstract}

\section{Introduction} A zero-one matrix $A$ contains another zero-one matrix $P$ (the pattern) if some sub-matrix of $A$ can be transformed to $P$ by changing some ones to zeros. $A$ avoids $P$ if $A$ does not contain $P$. A central problem about pattern avoidance is to find the extremal function $ex(n, P)$ which is the maximal number of ones in an $n$ by $n$ zero-one matrix that avoids the pattern $P$. The extremal functions of various patterns have been studied in many papers, including \cite{Fu}, \cite{Mi} and \cite{FH}. 

In this paper, we study the strategies of the pattern avoidance game which is related to the extremal function problem. Starting with an $n$ by $n$ all-zero matrix $A$, two players, Player $1$ and Player $2$, take turns changing zeros to ones (Player $1$ goes first). If any player's turn causes the matrix to contain the pattern $P$, then that player loses. We denote this game by PAG$(n,P)$. Each turn the player must change exactly one zero entry to an one. We only study cases where no dimension of $P$ exceeds $n$; other cases are trivial because in those cases the matrix can never contain $P$.

In Section $2$, we will study the strategies of the game for some special patterns $P$. In Section $3$, we will generalize the game to more than two players and higher dimensional matrix. The last section will be about open problems related to this pattern avoidance game. 

\section{Strategies for Special Patterns}
We start with the following theorem about the strategies when the pattern $P$ is the $1$ by $k$ matrix with all ones.

\begin{Th}
	Suppose $P$ is the $1$ by $k$ matrix with all ones. When two players play PAG$(n,P)$, the winner is Player $1$ if $k$ is even and $n$ is odd and is Player $2$ otherwise.
\end{Th}
 
\begin{proof}
	The key observation is that any player can change a zero in some row of $A$ to an one without losing the game that turn if and only if the number of ones in that row is less than $k-1$. That means there will be a time when each row has exactly $k-1$ ones and the next player loses no matter what their strategies are. At that time a total number of $(k-1)n$ zeros have been changed to ones. So if $k$ is even and $n$ is odd, then $(k-1)n$ is also odd and the next player (who will lose)is Player $2$. Otherwise Player $2$ wins.  
\end{proof}

The above theorem deals with all cases where $P$ is one dimensional. Let's move on to the next simplest case where $P$ is the $2$ by $2$ identity matrix.

\begin{Th} 
	Suppose $P$ is the $2$ by $2$ identity matrix. When two players play PAG$(n,P)$, the winner is Player $1$.
\end{Th}

\begin{proof}
	Here is the winning strategy for Player $1$: on the first turn, Player $1$ changes the zero in the top left corner of $A$ to an one. Then no one should pick zeros that are not in the top row or the leftmost column, because such a move will make the person the loser. There are $2(n-1)$ available choices left, and when they are all changed to ones, it is Player $2$'s turn. So Player $2$ is the loser.
\end{proof}

Another way to think about the above strategy for Player $1$ is the following: Player $1$ first picks the top left corner, and then he mimics Player $2$'s move. If Player $2$ picks a zero in the top row, then in the next turn Player $1$ picks the corresponding zero in the leftmost column. If Player $2$ picks a zero in the leftmost column, then Player $1$ picks the corresponding zero in the top row. 

The next theorem further explores this mimicking strategy for some other patterns.

\begin{Th}
	If $n$ is even and $P$ is composed of an odd number of identical rows or columns, then the winner is Player $2$.
\end{Th} 

\begin{proof}
	WLOG we assume that $P$ has $k$ identical columns and $k$ is an odd number. Since $n$ is even, Player $2$ can always make sure that after his turn the $2k-1$-th column of the matrix $A$ is always the same as the $2k$-th column for all $0<k<n/2+1$. We can prove that under this strategy Player $2$ is not the loser (hence Player $1$ is) by showing that if after a Player $2$'s turn the matrix contains the pattern $P$, then it already did before that turn. 
	
	Given Player $2$'s strategy, we know that after Player $2$'s turn the $2k-1$-th column of the matrix is the same as the $2k$-th column for all $0<k<n/2+1$. If the matrix contains $P$, it means we can find $k$ columns of 
	$A$ such that the sub-matrix of $A$ consisting of those $k$ columns contains the pattern $P$. But $P$ has an odd number of identical columns, so we can find $k+1$ columns of $A$ such that sub-matrix formed by any $k$ of them contains the pattern $P$. In that turn Player $2$ has only changed one zero in some column, so the rest of the $k$ columns contain the pattern $P$ and they are not changed in Player $2$'s turn. This means before Player $2$'s turn the matrix $A$ already contained the pattern $P$. Hence Player $2$ is never the loser.    
	 
\end{proof}

\section{Multi-player Games}
In this section we study the multi-player version of this pattern avoidance game. The setting is pretty similar: more than two players take turns changing zeros to ones until after someone's turn the matrix contains the pattern $P$, and that person is the loser. Other players win. 

In general the strategies for the multi-player game could be very hard to study, but for some special $P$ the "maximal" matrices all have the same number of ones, making things much simpler. We begin with the definition of "maximal" matrix. 

\begin{Def}
	Fixed a pattern $P$. A zero-one matrix $A$ is called maximal if $A$ does not contain $P$ and $A$ will contain $P$ as long as we change any zero of $A$ to one. 
\end{Def}

For any pattern $P$, there are many maximal matrices and those might not have the same number of ones. However, our next theorem shows that when $P$ is the $2$ by $2$ identity matrix, all maximal matrices have the same number of ones.

\begin{Th}
	Suppose $P$ is the $2$ by $2$ identity matrix and $A$ is an $h$ by $w$ maximal matrix for the pattern $P$. Then the number of ones in $A$ is $h+w-1$.  
\end{Th}

\begin{proof}
	First we claim that every column of $A$ has at least an one entry. To see this, take any entry $A_{i,j}$, there cannot be an one entry both to its upper left and lower right direction. If there's an one to its upper left, decrease $i$ until there isn't -- now there's an one to its left, and none to its lower right, so $A_{i,j}$ should be one.
	
	For $1\leq i\leq w$, define $h_i$ as the smallest row index such that $A_{h_i,i}=1$. The process above shows that $\{h_i\}$ is a non-increasing sequence. Moreover, for each column $i>1$ there's no reason any of $A_{h_i,i}, A_{h_i+1,i},\ldots,A_{h_{i-1},i}$ is not one. Similarly $A_{h_1,1},\ldots,A_{h,1}$ should all be one. So each column of $A$ has consecutive ones overlapping with its previous column at exactly one position. The total number of ones is thus $w+h-1$.
\end{proof}
	The above theorem shows that when $P$ is the $2$ by $2$ identity matrix, the loser of the multi-player pattern avoidance game is always the player who takes the $(h+w)$-th turn. 
	
	Now we can generalize the above theorem to $d$-dimensional matrix. This theorem was also included in the paper \cite{shenfu}.
	
\begin{Th}\label{mainthm}
    	Suppose $P$ is the $2\times 2\times \dots \times 2$ $d$-dimensional identity matrix($P$ has two entries that are one and $2^d-2$ entries that are zero, $d\geq 2$ ). $A$ is a $w_1\times w_2 \times \dots \times w_d$ $d$-dimensional maximal matrix for the pattern $P$. Then the number of entries of $A$ that are one is $\prod_{i=1}^dw_i-\prod_{i=1}^d(w_i-1)$.  
\end{Th}

\begin{proof}
	Let $M$ be a maximal $w_1\times\ldots\times w_d$ matrix that avoids $P$. We focus on the $d$-rows of $M$.
	
	A $d$-row $x=(x_1,\ldots,x_{d-1})$ is an {\it ancestor} of $d$-row $y=(y_1,\ldots,y_{d-1})$ if $x_i<y_i$ for $1\leq i\leq d-1$. We say $y$ is a {\it descendant} of $x$. $A(x)$ and $D(x)$ are the sets of ancestors and descendants of $x$, respectively. Define $M(x,x_d)=M_{x_1,\ldots,x_d}$. Ancestor and descendant relations are transitive, and the graphs of these two relations are directed acyclic.
	
	Similar to the proof of the $2$-dimensional case, we first claim that every $d$-row of $M$ has some one entry. Given a $d$-row $x$, if $A(x)$ is empty, then $M(x,w_d)=1$. If $D(x)$ is empty, then $M(x,1)=1$. If both sets are not empty, $h(A(x))$, the smallest $d$-coordinate of one entries in $A(x)$ must be greater than or equal to $l(D(x))$, the largest $d$-coordinate of one entries in $D(x)$, otherwise $M$ contains $P$. Then we pick some $y$ in $[l(D(x)), h(A(x))]$ and having $M(x,y)=1$ does not make $M$ contain $P$.
	
	By changing everything between two one entries in the same $d$-row to one, we can assume that the set $\{y:M(x,y)=1\}$ forms an integer arithmetic progressive sequence with common difference 1 for any $x$.
	
	For a $d$-row $x$, define $h(x)$ and $l(x)$ as the minimum and maximum $y$ such that $M(x,y)=1$, respectively. $w(M)$, the number of one entries in $M$, is then $\sum_{x}(l(x)-h(x)+1)$.
	
	Then let's examine $l(x)$. If $l(x)>h(z)$ for some $z\in A(x)$, then $M$ contains $P$. On the other hand, if $l(x)<\min_{z\in A(x)}h(z)$, setting $M(x,\min_{z\in A(x)}h(z))=1$ doesn't make $M$ contain $P$. Thus we conclude that $l(x)=\min(w_d, \min_{z\in A(x)}h(z))$. Similarly we have $h(x)=\max(1, \max_{z\in D(x)}l(z))$.
	
	We're now going to prove the theorem by induction on $d$ and $w_d$ through the following lemmas. Assume the theorem holds with $1,2,\ldots,d-1$ ($d \geq 3$)dimensions, and for $(w_1,\ldots,w_{d-1},1),(w_1,\ldots,w_{d-1},2),\ldots,(w_1,\ldots,w_{d-1},w_d-1)$.
	
	\begin{Lemma}\label{eq}
		A matrix $M$ avoiding $P$ is maximal if $l(x)=\min(w_d, \min_{z\in A(x)}h(z))$ and $h(x)=\max(1, \max_{z\in D(x)}l(z))$ for every $d$-row $x$.
	\end{Lemma}	
	
	\begin{proof}
		Clearly, if there's an one entry in some $d$-row $x$ with $d$-coordinate
		$y$
		bigger than $\min(w_d, \min_{z\in A(x)}l(z))$, then $y$ either goes
		beyond $w_d$, which is impossible, or $y>l(z)$ for some ancestor $z$ of $x$, and $M$ contains
		$P$. Similarly we couldn't have any one entry in any $d$-row $x$ with
		$d$-coordinate less than $\max(1, \max_{z\in D(x)}h(z))$.
	\end{proof}
	
	We say that a $d$-row $x=(x_1,\ldots,x_{d-1})$ is a {\it semi-ancestor} of $d$-row $y=(y_1,\ldots,y_{d-1})$ if $x_i\leq y_i$ for all $i$ in $[1,d-1]$ and $x_i<y_i$ for some $i$ in $[1,d-1]$.
	
	\begin{Lemma} If the set $\{x:l(x)=w_d, A(x)\neq\emptyset\}\neq\emptyset$, there are $d$-rows $x$ and $x'$, such that $x=(x'_1+1,\ldots,x'_{d-1}+1)$, $l(x)=w_d>h(x)$, and $x'$ has no other descendant $z$ with $l(z)=w_d$.
	\end{Lemma}
	Proof: In the first stage we find a $d$-row $x$ with $l(x)=w_d>h(x)$ and $A(x)\neq\emptyset$. Starting from any $d$-row $x$ with $l(x)=w_d$, if $h(x)=w_d$, then there is descendant $z$ of $x$ with $l(z)=w_d$. We move from $x$ to $z$ and see if $h(z)<w_d$. We repeat the process until we find a $x''$ with $l(x'')=w_d>h(x'')$, which is guaranteed to have non-empty ancestor set $A(x'')$. Such $x''$ exists because there are only finite number of $d$-rows, and whenever we hit one without descendants, $h(x'')=1$. We set $x=x''$ and $x'=(x_1-1,\ldots,x_{d-1}-1)$.
	
	In the second stage, if $x'$ has no other descendant $z$ with $l(z)=w_d$ then we're done. If it does, $D(z)\subset D(x)$ so $h(z)\leq h(x)<w_d$, i.e. $z$ qualifies the requirement for $x$ too. Moreover $x$ is a semi-ancestor of $z$, so we reset $x$ to $z$. We repeat this process again and again until $x'$ has no other descendant $z$ with $l(z)=w_d$. This process terminates because in the sequence of $x$ we found, each element is a semi-ancestor of it subsequent element, and this sequence couldn't be infinitely long.
	
	\begin{Lemma}\label{convert}
		 If the set $\{x:l(x)=w_d, A(x)\neq\emptyset\}\neq\emptyset$, we can convert $M$ to another maximal matrix $M'$ that avoids $P$ with the set smaller and yet $w(M)=w(M')$.
	\end{Lemma}
	
	\begin{proof}
	Proof: We find $d$-rows $x$ and $x'$, such that $x=(x'_1+1,\ldots,x'_{d-1}+1)$, $l(x)=w_d>h(x)$, and $x'$ has no other descendant $z$ with $l(z)=w_d$.
	
	Construct $M'$ identical to $M$ except $l'(x)=w_d-1$ and $h'(x')=w_d-1$. $M'$ has fewer $d$-rows in the set $\{x:l'(x)=w_d, A(x)\neq\emptyset\}$, plus $w(M')=w(M)$.
	
	We now show that $M'$ is still a maximal matrix that avoids $P$. The only changes we made is decreasing $l(x)$ and $h(x')$ by $1$. Clearly equation in Lemma \ref{eq} is not violated by $M'$ because $x'$ is an ancestor of $x$ and $x'$ has no other descendant $z$ with $l(z)=w_d$. So we're done.
	
	\end{proof}
	
	\begin{Lemma}\label{final}
	If $\{x:l(x)=w_d,A(x)\neq\emptyset\}=\emptyset$, then $w(M)=\prod_{i=1}^dw_i-\prod_{i=1}^d(w_i-1)$.
	\end{Lemma}
	
	\begin{proof}Transform $M$ to $M'$ by removing the $w_d$-th $d$-cross section. Since $l(x)=w_d$, if $A(x)=\emptyset$, then $w(M')=w(M)-|\{x:A(x)=\emptyset\}|$.
	By assumption,$w(M')=w_1w_2\ldots w_{d-1}(w_d-1)-(w_1-1)(w_2-1)\ldots (w_{d-1}-1)(w_d-2)$.
	$A(x)=\emptyset$ if and only if at least one of the coordinates $x_1,\ldots,x_{d-1}$ is $1$. So $|\{x:A(x)=\emptyset\}|=w_1\ldots w_{d-1}-(w_1-1)\ldots (w_{d-1}-1)$. Adding them together, $w(M)=\prod_{i=1}^dw_i-\prod_{i=1}^d(w_i-1)$.
	
	\end{proof}
	
	Finally, given any matrix $M$ that maximally avoids $P$, if necessary by applying Lemma \ref{convert} a finite number of times we can convert it to another maximal $M'$ with the same number of one entries, still avoids $P$, and with empty $\{x:l(x)=w_d,A(x)\neq\emptyset\}$. And then by Lemma \ref{final} we have $w(M)=\prod_{i=1}^dw_i-\prod_{i=1}^d(w_i-1)$. This concludes our proof.
	
\end{proof}  

\section{Open Problems}
\begin{itemize}
	\item Determine the winner of PAG$(n,P)$ for some other patterns $P$.
	\item Are there other patterns $P$ such that the maximal matrices always have the same number of one entries?
\end{itemize}

\section*{Acknowledgement}

CrowdMath is an open program created by the MIT Program for Research in Math, Engineering, and Science (PRIMES) and Art of Problem Solving that gives high school and college students all over the world the opportunity to collaborate on a research project. The 2016 CrowdMath project is online at http://www.artofproblemsolving.com/polymath/mitprimes2016. Theorem \ref{mainthm} was also included in the paper \cite{shenfu}.


\begin{thebibliography}{99} 

\bibitem{Fu} Fulek, Radoslav. "Linear bound on extremal functions of some forbidden patterns in 0?1 matrices." Discrete Mathematics 309.6 (2009): 1736-1739.

\bibitem{Mi} Mitchell, J. Shortest rectilinear paths among obstacles. Cornell University Operations Research and Industrial Engineering, 1987.

\bibitem{FH} F$\ddot{\text{u}}$redi, Zolt$\acute{\text{a}}$n, and P$\acute{\text{e}}$ter Hajnal. "Davenport-Schinzel theory of matrices." Discrete Mathematics 103.3 (1992): 233-251.

\bibitem{shenfu} Tsai, Shen-Fu. On equivalence between maximal and maximum antichains over rangerestricted vectors with integral coordinates and the number of such maximal antichains. https://arxiv.org/abs/1701.06750.

\end{thebibliography}
\end{document}